\documentclass[a4paper,10pt]{article}

\usepackage[ansinew]{inputenc}
\usepackage{amssymb, amsthm}
\usepackage{latexsym}
\usepackage{amsmath}
\usepackage{hyperref}
\usepackage{authblk}

\newcommand{\und}{\qquad\text{and}\qquad}
\newcommand{\sfrac}[2]{{\textstyle\frac{#1}{#2}}}
\renewcommand{\=}{\ =\ }
\newcommand{\tr}{\text{tr}}

\newtheorem{theo}{Theorem}[section]

\newtheorem{lem}[theo]{Lemma}
\theoremstyle{definition}

\makeatletter
\renewcommand{\@maketitle}{
\newpage
\null
\vskip 1em%
\begin{flushright}
ITP--UH--12/10\\
\end{flushright}
\begin{center}%
{\Large \textbf \@title \par}%
\vskip 2.5em 
{\large \@author \par}
\vskip 1.5 em
\end{center}%
\par} \makeatother

\title{Heterotic compactifications on nearly K\"ahler manifolds}
\author{
Olaf Lechtenfeld$^{\dagger\times}$, \
Christoph N\"olle$^\dagger$ \ and \ 
Alexander D. Popov$^\ast$}
\date{}

\begin{document}
\setcounter{page}{0}
\parindent=0cm
\maketitle
\thispagestyle{empty}

\begin{center} 
${}^\dagger${\em 
Institut f\"ur Theoretische Physik, Leibniz Universit\"at Hannover \\
Appelstra{\ss}e 2, 30167 Hannover, Germany } \\
email: lechtenf, noelle@itp.uni-hannover.de

\bigskip

${}^\times${\em
Centre for Quantum Engineering and Space-Time Research \\
Leibniz Universit\"at Hannover \\
Welfengarten 1, 30167 Hannover, Germany }

\bigskip

$^\ast${\em Bogoliubov Laboratory of Theoretical Physics, JINR \\
141980 Dubna, Moscow Region, Russia} \\
email: popov@theor.jinr.ru
\end{center}

\bigskip

\begin{abstract} 
We consider compactifications of heterotic supergravity on anti-de Sitter 
space, with a six-dimensional nearly K\"ahler manifold as the internal space. 
Completing the model proposed by Frey and Lippert~\cite{Frey_Lippert} with the
particular choice of SU(3)/U(1)$\times$U(1) for the internal manifold, we show
that it satisfies not only the supersymmetry constraints but also the equations
of motion with string corrections of order~$\alpha'$. Furthermore, we present 
a non-supersymmetric model. In both solutions we find confirmed a recent result
of Ivanov~\cite{Ivanov} on the connection used for anomaly cancellation. 
Interestingly, the volume of the internal space is fixed by the supersymmetry 
constraints and/or the equations of motion.
\end{abstract}

\bigskip

\tableofcontents

\numberwithin{equation}{section}

\section{Introduction}
 There has been considerable interest in heterotic flux compactifications in recent years. Compared to the traditional Calabi-Yau models with their problematic moduli, these might include more realistic examples where all moduli are fixed by a steep effective potential. A pioneering work in this direction was Strominger's analysis of supersymmetric compactifications with only bosonic fields non-vanishing \cite{Strominger}, leading to his well-known set of conditions on the internal manifold and the fields. Although it has been known for quite some time now that these conditions admit non-trivial solutions \cite{Li:2004hx}, explicit models are rare \cite{FuYau08, FIUV, Ugarte09}. 

 It has been proposed in \cite{Govindarajan, Frey_Lippert} to consider more general vacua with external anti-de Sitter space, where the gaugino is not supposed to vanish. Similarly to Strominger's approach Frey and Lippert use the condition for supersymmetric vacua to derive constraints on the internal manifold \cite{Frey_Lippert}.
 Among their solutions the simplest models are given by a direct product of four-dimensional anti-de~Sitter space and a six-dimensional compact nearly K\"ahler manifold. Nearly K\"ahler manifolds carry an SU(3) structure \cite{Friedrich02, CCDLMZ} with totally antisymmetric intrinsic torsion, giving rise to a three-form. The latter makes them promising candidates for compactification models where one can take the $H$-field equal to the intrinsic torsion. This relation follows indeed from the supersymmetry conditions. Furthermore, it is suspected that nearly K\"ahler compactifications have few moduli; e.g.\ only four compact six-dimensional nearly K\"ahler spaces are known. 

 The model proposed by Frey and Lippert is incomplete, since their discussion of the Bianchi identity assumes a certain form of the gauge field, which does not have the instanton property necessary for the vanishing of the gaugino variation. This problem has already been adressed by Manousselis, Prezas and Zoupanos \cite{Manousselis:2005xa}, who also include a dilatino condensate in their considerations, to have more freedom to adjust parameters. They showed that there exist complete supersymmetric solutions for all of the four nearly K\"ahler spaces. 

 Anomaly cancellation in heterotic string theory requires the introduction of a connection $\tilde\Gamma$ on the tangent bundle. One usually takes the Levi-Civit\`a connection plus a torsion given by the $H$-field, yet other choice are possible. For bosonic Minkowski compactifications, Ivanov has shown recently that the choice of $\tilde\Gamma$ crucially determines whether or not the equations of motion follow from the supersymmetry conditions and the Bianchi identity for~$H$~\cite{Ivanov}. We find this result confirmed: supersymmetric solutions fulfil the dilaton equation only if $\tilde\Gamma$ is the instanton connection on the internal space. 
 
In this paper, we construct two complete solutions to the heterotic supergravity equations with string corrections of order $\alpha'$, on the space AdS${}_4(r)\times K(\rho)$, where $K(\rho)$ is the nearly K\"ahler coset SU(3)/U(1)$\times$U(1) with scale~$\rho$. We restrict ourselves to a vanishing gravitino and dilatino,\footnote{In the equations of motion, we also put the dilaton to zero.} which has the advantage that we are able to check explicitly whether the equations of motion are satisfied. Although generally one expects the vanishing of the fermionic supersymmetry variations to imply the equations of motion, this is not always the case. Our first solution is supersymmetric and features a gaugino condensate, while our second one is non-supersymmetric and without gaugino.

After presenting the action, supersymmetry transformation and field equations in Section \ref{sec:ActionEtc}, we find the AdS radius~$r$ fixed in terms of~$\alpha'$. 
The rich differential geometry of homogeneous nearly K\"ahler spaces is reviewed in Section \ref{sec:NKmfs}, where we also calculate all the quantities appearing in the equations throughout this paper.
In Section \ref{sec:susySols} we combine the constraints on the external and internal geometries and show that at least the space SU(3)/U(1)$\times$U(1) supports supersymmetric models, which are not contained in the solutions of~\cite{Manousselis:2005xa}. 
In Section \ref{sec:nonSUSYNKCptf} we present a non-supersymmetric solution to all equations of motion (and the Bianchi identity), again on the coset SU(3)/U(1)$\times$U(1), with vanishing gaugino for simplicity. Again for $\tilde\Gamma$ we need to choose the instanton connection. In both solutions, the gauge connection is proportional to the spin connection, and the scale $\rho$ of the internal manifold is determined by~$\alpha'$.
In the Conclusions we combine our results and display the explicit form of 
our two solutions.
Attempts to solve the equations for the alternative $K$\,=\,SU(2)${}^3$/SU(2) require a different choice for $\tilde\Gamma$, which fails to satisfy the dilaton equation.

\newpage

\section{Action, supersymmetry and field equations}\label{sec:ActionEtc}
\paragraph{Field content.\ }
     The low-energy field theory limit of heterotic string theory is given by $d{=}10$, ${\cal N}{=}1$ supergravity coupled to a super-Yang-Mills multiplet. It
 contains the following fields living on the 10$d$ spacetime $M$ and
transforming in particular irreps of the tangent SO(9,1)~\cite{GSW}:\\[-20pt]
\begin{itemize}
\addtolength{\itemsep}{-6pt}
\item graviton $g$, a metric on $M$, in a \underline{54}
\item dilaton $\phi$, a function on $M$, in a \underline{1}
\item Kalb-Ramond field $B$, a two-form on $M$, in a \underline{45}
\item gauge field $A$, a one-form on $M$, in a \underline{10}
\item gravitino $\psi$, left-handed Majorana-Weyl vector-spinor, in a $\underline{144}_s$
\item dilatino $\lambda$, right-handed Majorana-Weyl spinor, in a $\underline{16}_s$
\item gaugino $\chi$, left-handed Majorana-Weyl spinor, in a $\underline{16}_c$
\end{itemize} 
\vspace{-8pt}
Furthermore there is a left-handed Majorana-Weyl supersymmetry generator $\varepsilon$. Besides these fundamental fields, we need the curvature forms
 \begin{equation}\label{CurvDefi}
  F\=dA + A\wedge  A \und
  H\=dB + \sfrac14\alpha'\big[\omega_{CS}(\tilde \Gamma)-\omega_{CS}(A)\big],
 \end{equation} 
 where $\omega_{CS}$ denotes the Chern-Simons-forms
 \begin{equation}
 \omega_{CS} (\tilde\Gamma) \= \text{tr} \big(\tilde R \wedge \tilde \Gamma - \sfrac 23\tilde  \Gamma\wedge \tilde \Gamma\wedge \tilde \Gamma\big) \und
 \omega_{CS} (A) \= \text{tr} \big(F \wedge A - \sfrac 23 A\wedge A\wedge A\big),
 \end{equation}
and $\tilde\Gamma$ is a connection on $TM$, whose choice is ambiguous.\footnote{
Different connections correspond to different regularization schemes in the 2$d$ sigma model, and are related by field redefinitions, see \cite{HullAnomalies} and the list of references for different choices in the introduction of \cite{FIUV}.
}
Basic choices are:\\[-20pt]
\begin{itemize}
\addtolength{\itemsep}{-6pt}
\item the Levi-Civit\`a connection $\Gamma^{\text{LC}}(g)$
\item the plus-connection $\Gamma^+=\Gamma^{\text{LC}}-\sfrac12H$
\item the minus-connection $\Gamma^-=\Gamma^{\text{LC}}+\sfrac12H$
\item the Chern connection $\Gamma^{\text{Ch}}$ \cite{Strominger}.
\end{itemize}
\vspace{-8pt}
Finally, there appears the space-time curvature form
\begin{equation}
\tilde R \= d\tilde\Gamma + \tilde\Gamma\wedge\tilde\Gamma,
\end{equation}
for whichever connection has been chosen.
Traces are taken over the adjoint representation of the gauge group or 
of SO(9,1), depending on the context.

\paragraph{Action and supersymmetry.\ }
In this paper we are putting the gravitino and dilatino to zero,
\begin{equation}
\psi = 0 \und \lambda = 0 .
\end{equation}
With this, the low-energy action up to and including terms of order $\alpha'$
reads~\cite{Frey_Lippert, BergshdR}
 \begin{equation}
\mathcal S(g,\phi,B,\chi,A)\=\int_M\!\!d^{10}x\sqrt{\det g}\ e^{-2\phi}\Big\{  
\text{Scal} +4 |d\phi|^2 -\sfrac 12 |T|^2 + \sfrac14\alpha'\text{tr}\big( 
|\tilde R|^2-|F|^2 -2 \overline\chi \mathcal D \chi\big)\Big\}, 
\end{equation}
where 
\begin{equation}
T\=H-\sfrac 12 \Sigma \qquad\text{with}\qquad
\Sigma \= \sfrac1{24}\alpha'\,
\text{tr} (\overline\chi\gamma_M\gamma_N\gamma_P\chi)\ 
dx^M\wedge dx^N\wedge dx^P
\end{equation}
and
\begin{equation}
\text{tr}|\tilde R|^2 \= \sfrac 12 \tilde R_{MNPQ}\tilde R^{MNPQ}   \und
\text{tr}|F|^2 \= \sfrac12 \text{tr} F_{MN}F^{MN}
\end{equation}
Finally, $\mathcal D=\gamma^M\nabla_M$ denotes the Dirac operator, 
coupled to $\Gamma^{\text{LC}}(g)$ and to~$A$.
The action is invariant under ${\cal N}{=}1$ supersymmetry transformations \cite{BergshdR}, which act on the fermions as 
  \begin{align}\label{gauginoSusyVar}
    \delta \psi_M &\= \nabla_M \varepsilon -\sfrac 18H_{MNP}\gamma^N\gamma^P\varepsilon + \sfrac 1{96}\gamma(\Sigma)\gamma_M\varepsilon  , \nonumber\\[4pt]
    \delta \lambda &\= -\sfrac 12 \gamma\big(d\phi -\sfrac 1{12}H- \sfrac 1{48}\Sigma\big)\varepsilon ,\\[4pt]
    \delta \chi &\= -\sfrac 14 \gamma(F) \varepsilon, \nonumber
 \end{align}
where $\gamma$ denotes the map from forms to the Clifford algebra,
\begin{equation}
 \gamma \big(\sfrac 1{p!}\omega_{M_1\dots M_p}\ dx^{M_1}\wedge\dots\wedge dx^{M_p}\big)\= \omega_{M_1\dots M_p} \gamma^{M_1} \dots\gamma^{M_p}.
\end{equation}

\paragraph{Field equations.\ }
The equations of motion take the form (we symmetrize with weight one)
\begin{equation} \label{eom}
\begin{aligned}
\text{Ric}_{MN} +2 (\nabla d\phi)_{MN} 
-\sfrac 18 T_{PQ(M}{H_{N)}}^{PQ}
+\sfrac14\alpha' \Big[ \tilde R_{MPQR}{\tilde R_N}^{\ PQR} 
- \text{tr}\big(F_{MP} {F_N}^P
-\sfrac 12\overline \chi\gamma_{(M} \nabla_{N)}\chi\big)\Big]&=0, \\
\text{Scal} + 4 \Delta \phi  -4|d\phi|^2 -\sfrac 12| T|^2
+\sfrac14\alpha'\text{tr}\Big[ |\tilde R|^2 - |F|^2 - 2\overline\chi\mathcal D\chi\Big]&=0 ,\\[2pt]
\big( \mathcal D - \sfrac 1{24} \gamma(T) \big)e^{-2\phi}\chi&=0, \\[4pt]
e^{2\phi }d\ast(e^{-2\phi}F) + A\wedge\ast F- \ast F\wedge A+\ast T\wedge F&=0, \\[4pt]
d\ast e^{-2\phi}T &=0. 
\end{aligned}
\end{equation}
 The derivation of the equations is greatly simplified by a Lemma in \cite{BergshdR}, as was pointed out by Becker and Sethi \cite{BeckerSethi}. It implies that up to this order in $\alpha'$ one can neglect variations of the form $\frac{\partial S}{\partial \tilde \Gamma}\frac {\partial\tilde \Gamma}{\partial(\cdots)}$, for any field~$(\cdots)$.
 Besides these equations, the Bianchi identity for $H$ must be satisfied, which follows from the definition \eqref{CurvDefi}:
 \begin{equation}\label{HBianchi}
  dH\=\sfrac14\alpha' \text{tr}[\tilde R\wedge \tilde R - F\wedge F].
 \end{equation} 
 One expects the vanishing of the fermionic supersymmetry transformations plus the Bianchi identity to imply all equations of motion, but whether or not this is true depends on the connection $\tilde\Gamma$.\footnote{
 Ivanov~\cite{Ivanov} proves that solutions of Strominger's equations (which are equivalent to the vanishing of \eqref{gauginoSusyVar} with vanishing gaugino) satisfy the equations of motion only for $\tilde\Gamma=\Gamma^-$.}

Taking the trace of the Einstein equation (the first one in~(\ref{eom})) gives
 \begin{equation}
  \text{Scal} + 2\Delta \phi -\sfrac 32 (T,H)
   + \sfrac12\alpha' \text{tr}\big[|\tilde R|^2-|F|^2  -\sfrac 12 \overline \chi\mathcal D\chi \big] \=0
  \qquad\text{with}\quad (T,H) \equiv\sfrac 1{3!}T_{MNP}H^{MNP}.
 \end{equation}
This equation can be combined with 
the dilaton equation (the second one in~(\ref{eom})) in two different ways as
\begin{equation}\label{EinsteinDilatonCombiGaugino}
\begin{aligned}
e^{2\phi}\Delta e^{-2\phi} +\sfrac 12 |T|^2 -\sfrac 32 (T,H) +\sfrac14\alpha' \text{tr}\big[ |\tilde R|^2 - |F|^2+\overline\chi\mathcal D\chi\big] &\=0,\\[4pt]
\text{Scal} -\sfrac 92 e^{ 4\phi/3 }\Delta e^{-4\phi/3} -|T|^2+ \sfrac 32 (T,H)-\sfrac34\alpha' \text{tr}(\overline \chi\mathcal D\chi)&\= 0.
\end{aligned}
\end{equation}
 One can replace the dilaton equation by one these, and if further the Einstein equation is satisfied, then the other equation in \eqref{EinsteinDilatonCombiGaugino} is implied. 
For the remainder of the paper we put the dilaton to zero,
\begin{equation}
\phi = 0 ,
\end{equation}
motivated by the results in \cite{Frey_Lippert}. 
This noticeably simplifies the equations of motion~(\ref{eom}).

\paragraph{Space-time factorization.\ }
Our interest is in space-time manifolds of direct product form,
\begin{equation}
M \= \text{AdS}_4(r)\times K,
\end{equation}
with a 4$d$ anti-de Sitter space of `radius'~$r$ as `external' factor 
and a 6$d$ compact Riemannian `internal' space~$K$. 
Small Greek indices shall be restricted to the external part, while
small Latin indices will be reserved for the internal dimensions.
Furthermore, we assume that $F$, $H$ and $\Sigma$ are restricted to~$K$, 
i.e.\ they do not depend on the AdS~coordinates. The components of $\tilde R$ 
in AdS direction are taken to coincide with the Riemann curvature of AdS.
This further simplifies the equations.
{}From now on, hatted quantities refer to the AdS~part, 
and unhatted ones live on~$K$. The ambiguity in picking $\tilde\Gamma$
is a choice of connection on~$K$.

As a consequence of the splitting, the equations of motion (\ref{eom}) 
decompose. The Einstein equation (first in (\ref{eom})) splits into 
\begin{equation} \label{Einsteinsplit}
\begin{aligned}
\widehat{\text{Ric}}_{\mu\nu} +\sfrac14\alpha' 
\widehat{R}_{\mu\alpha\beta\gamma}{\widehat{R}_\nu}^{\ \alpha\beta\gamma} &\=
\sfrac18\alpha'\text{tr} \big(\overline \chi \widehat\gamma_{(\mu}\widehat\nabla_{\nu)}\chi\big), \\[4pt]
 0&\= \sfrac 18\alpha' \text{tr}\big(\overline\chi \big(\widehat \gamma_\mu \nabla_a + \gamma_a\widehat \nabla_\mu\big)\chi\big)\\[4pt]
\text{Ric}_{ab} -\sfrac 18 T_{cd(a}{H_{b)}}^{cd}
+\sfrac14\alpha' \Big[ \tilde R_{acde}{\tilde R_b}^{\ cde} 
- \text{tr}\big(F_{ac} {F_b}^c\big)\Big] &\=
\sfrac18\alpha'\text{tr}\big(\overline \chi \gamma_{(a}\nabla_{b)}\chi\big), 
\end{aligned}
\end{equation}
and the two combinations~(\ref{EinsteinDilatonCombiGaugino}) then read
\begin{equation}
\begin{aligned}
\sfrac12|T|^2 -\sfrac 32 (T,H) 
+\sfrac14\alpha' \text{tr}\Big[ |\widehat{R}|^2 + |\tilde R|^2 - |F|^2 \Big]&\=
-\sfrac14\alpha' \text{tr}\big({\overline\chi}\ 
(\widehat{\mathcal D}+\mathcal D)\ \chi \big), \\[4pt]
\widehat{\text{Scal}} + \text{Scal}  -|T|^2
+ \sfrac 32 (T,H) &\= +\sfrac34\alpha' \text{tr}
\big({\overline\chi}\
(\widehat{\mathcal D}+\mathcal D)\ \chi\big) .
\end{aligned}
\end{equation}
 The gaugino, Yang-Mills and Kalb-Ramond equations become
\begin{equation} \label{gauginoYMKR}
\begin{aligned}
\big(\widehat{\mathcal D}+\mathcal D-\sfrac1{24}\gamma(T)\big)
 \chi&\=0,\\[4pt]
d\ast F + A\wedge\ast F- \ast F\wedge A+\ast T\wedge F&\=0, \\[4pt]
d\ast T &\=0 .
\end{aligned}
\end{equation}
The gravitational data on $\text{AdS}_4(r)$ are
\begin{equation}
\widehat{\text{Scal}}=-\frac {12}{r^2}, \qquad
\widehat{\text{Ric}}= \sfrac14\widehat{\text{Scal}}\ \widehat{g} =
-\frac{3}{r^2}\,\widehat{g}, \qquad
\widehat{R}_{\mu\alpha\beta\gamma}{\widehat{R}_\nu}^{\alpha\beta\gamma} = 
\sfrac1{24}\widehat{\text{Scal}}^2 \widehat{g}_{\mu\nu} 
=\frac {6}{r^4}\, \widehat{g}_{\mu\nu}.
\end{equation}
The gaugino is taken to factorize as
\begin{equation} \label{chifactor}
\chi \= \widehat\chi\otimes\eta + \widehat\chi^*\otimes\eta^*,
\end{equation}
where $\widehat\chi$ is an anticommuting spinor on $\text{AdS}_4$ 
with values in the adjoint of the gauge group, while $\eta$ denotes 
a commuting spinor on~$K$, which we assume to be normalized: $\overline \eta \eta=1$. On $\widehat \chi$ we impose the massless Dirac equation
on $\text{AdS}_4$~\cite{Bachelot07}, 
\begin{equation}
\widehat{\mathcal D}\,\widehat\chi \=0,
\end{equation}
so that the gaugino equation (first in (\ref{gauginoYMKR})) implies
\begin{equation} \label{trgaugino}
\big(\mathcal D -\sfrac1{24}\gamma(T)\big)\,\eta\=0 \und
\tr(\overline\chi\mathcal D\chi)\=\sfrac 1{\alpha'}(T,\Sigma).
\end{equation}
Inserting these relations into the equations of motion, 
we obtain the conditions~\footnote{
Note that $\widehat{\overline\chi}=\widehat{\chi}^\dagger\gamma^0$
but $\overline\eta=\eta^\dagger$.}
\begin{equation}\label{EOM_split}
\begin{aligned}
-\big(\sfrac3{r^2}-\sfrac3{2r^4}\alpha'\big)\,\widehat{g}_{\mu\nu}
&\=\sfrac18\alpha'\text{tr}\big(
\widehat{\overline\chi}\widehat\gamma_{(\mu}\widehat\nabla_{\nu)}\widehat\chi + \widehat{\overline\chi}^*\widehat\gamma_{(\mu}\widehat\nabla_{\nu)}\widehat\chi ^*
\big),\\[2pt]
0\=\text{tr}\Big( (\widehat{\overline\chi}\widehat\gamma_\mu\widehat\chi)
(\overline\eta\nabla_a\eta)+(\widehat{\overline{\chi}}^\ast \widehat\gamma_\mu \widehat\chi^\ast)(\overline\eta^\ast\nabla_a\eta^\ast) - 
(&\widehat{\overline\chi}\widehat\nabla_\mu\widehat\chi^\ast)(\overline\eta\gamma_a\eta^\ast) +(\widehat{\overline\chi}^\ast \widehat\nabla_\mu \widehat \chi)(\overline\eta^\ast \gamma_a\eta)\Big),\\
\text{Ric}_{ab} -\sfrac 18 T_{cd(a}{H_{b)}}^{cd}
+\sfrac14\alpha' \Big[ \tilde R_{acde}{\tilde R_b}^{\ cde} 
- \text{tr}\big(F_{ac} {F_b}^c\big)\Big] &\=
\sfrac18\alpha'\text{tr}\Big( 
(\widehat{\overline\chi}^\ast\widehat\chi)
(\overline\eta^\ast\gamma_{(a}\nabla_{b)}\eta)-
(\widehat{\overline\chi}\widehat\chi^\ast)
(\overline\eta\gamma_{(a}\nabla_{b)}\eta^\ast) \Big),\\
\sfrac12|T|^2 -\sfrac 32 (T,H) 
+\sfrac3{r^4}\alpha'+\sfrac14\alpha'\text{tr}\Big[ |\tilde R|^2-|F|^2 \Big]&\=
-\sfrac 14 (T,\Sigma),\\
-\sfrac{12}{r^2} + \text{Scal}  +\sfrac 12|T|^2&\= 0, \\[4pt]
\big(\mathcal D -\sfrac1{24}\gamma(T)\big) \eta&\=0,  \\[4pt]
 d\ast F + A\wedge\ast F- \ast F\wedge A+\ast T\wedge F&\=0, \\[4pt]
d\ast T&\=0 .
\end{aligned}
\end{equation}
which entangle the internal fields with the AdS data. 

On the right-hand side of these equations we encounter external gaugino 
bilinears, which are nilpotent on the classical level. The standard lore 
to give meaning to these terms performs a quantum average $\langle\dots\rangle$
over the space-time fermionic degrees of freedom.
At this stage, assumptions about the fermionic quantum correlators enter:
we assume the presence of a suitable space-time gaugino condensate as 
a backdrop for the bosonic equations, namely
\begin{equation} \label{condensate}
\langle\tr\,\widehat{\overline\chi}\widehat\chi^\ast\rangle \= i\,\Lambda^3
\qquad\text{but}\qquad
\langle\tr\,\widehat{\overline\chi}\widehat{M}\widehat\chi \rangle \=
\langle\tr\,\widehat{\overline\chi}\widehat{M}\widehat\chi^\ast \rangle \= 0
\end{equation}
for all non-scalar operators $\widehat{M}$,
to be consistent with (\ref{trgaugino}). 
The condensate scale~$\Lambda\in\mathbb R$ will be fixed later.
After averaging over the gaugino,
our set of equations (\ref{EOM_split}) simplifies to
\begin{equation} \label{EOM_condensate}
\begin{aligned}
-\big(\sfrac3{r^2}-\sfrac3{2r^4}\alpha'\big)\,\widehat{g}_{\mu\nu} &\=0,\\[2pt]
\text{Ric}_{ab} -\sfrac 18 T_{cd(a}{H_{b)}}^{cd}
+\sfrac14\alpha' \Big[ \tilde R_{acde}{\tilde R_b}^{\ cde} 
- \text{tr}\big(F_{ac} {F_b}^c\big)\Big] &\=
\sfrac{i}8\,\Lambda^3\alpha'\big(
\overline\eta^\ast\gamma_{(a} \nabla_{b)}\eta-
\overline\eta\,\gamma_{(a} \nabla_{b)}\eta^\ast \big),\\
\sfrac12|T|^2 -\sfrac 32 (T,H) 
+\sfrac3{r^4}\alpha'+\sfrac14\alpha'\text{tr}\Big[ |\tilde R|^2-|F|^2 \Big]&\=
-\sfrac 14 (T,\Sigma),\\
-\sfrac{12}{r^2} + \text{Scal}  +\sfrac 12|T|^2&\= 0, \\[4pt]
\big(\mathcal D -\sfrac1{24}\gamma(T)\big) \eta&\=0,  \\[4pt]
 d\ast F + A\wedge\ast F- \ast F\wedge A+\ast T\wedge F&\=0, \\[4pt]
d\ast T&\=0,
\end{aligned}
\end{equation}
where we continue to use the symbol $\Sigma$ for the condensate
\begin{equation}
\langle \Sigma \rangle \= \sfrac{i}{24}\,\Lambda^3\alpha'\big(
\overline\eta^\ast\gamma_a\gamma_b\gamma_c\eta -
\overline\eta\,\gamma_a\gamma_b\gamma_c\eta^\ast \big)\,
dx^a\wedge dx^b\wedge dx^c .
\end{equation}
Remarkably,
the first equation fixes the $\text{AdS}_4$ radius in terms of~$\alpha'$,
\begin{equation}
r^2 \= \sfrac12\alpha'.
\end{equation}
We note that tracing the second equation simplifies its right-hand side
to $\sfrac18(T,\Sigma)$, proportional to the right-hand side of the
third equation. In the fourth equation (the dilaton equation),
the negative contribution $\widehat{\text{Scal}}=-\frac{12}{r^2}$ 
allows for internal manifolds of positive scalar curvature, 
which are excluded in the usual Minkowski compactifications.
The final three equations are conditions on the commuting spinor~$\eta$,
the Yang-Mills connection~$A$ and the torsion~$H$ on the internal manifold~$K$.
We shall construct solutions to them, after having introduced the geometry of
nearly K\"ahler manifolds in the following section.

\section{Nearly K\"ahler manifolds}\label{sec:NKmfs}
\paragraph{Definitions.}\label{su3StructureDefis}
A six-dimensional Riemannian manfiold $(K,g)$ is said to carry an SU(3) 
structure, if there is a compatible almost-complex structure 
$J\in \Gamma$(End$(TK))$, with $J^2=-1$ and $g(JX,JY)=g(X,Y)$, 
the associated two-form $\omega=g(\cdot,J\cdot)\in \Omega^{(1,1)}(K)$, 
and a non-vanishing (3,0)-form $\Omega\in \Omega^{(3,0)} (K)$. 
Every SU(3)-manifold has a metric-compatible connection $\nabla^-$ with 
holonomy contained in SU(3). If it coincides with the Levi-Civit\`a connection 
$\nabla^{LC}$ one has a Calabi-Yau space, with $d\omega=d\Omega=0$. 
Furthermore, SU(3)-structure manifolds are spin and carry a nontrivial, 
covariantly constant spinor (w.r.t.\ the canonical connection) of each 
chirality, $\eta$ and $\eta^\ast$, which are charge conjugates of each other. 
A nearly K\"ahler manifold is characterized by the conditions 
\begin{equation}\label{NK_defRels}
d\omega \= -3\,\varsigma\, \text{Re}(\Omega),\qquad 
d\Omega \= 2i\varsigma\, \omega\wedge \omega
\qquad\text{with}\quad \varsigma\in\mathbb R.
\end{equation} 
These forms can be constructed from the spinors $\eta$ and $\eta^\ast$ as 
\begin{align}
\omega&\=\sfrac i2\,\overline\eta\,\gamma_a\gamma_b\eta\,e^a\wedge e^b,
\nonumber\\[4pt]
\Omega&\=+\sfrac 16\,\overline\eta\,\gamma_a\gamma_b\gamma_c\eta^\ast
e^a\wedge e^b\wedge e^c, \\[4pt]
\overline\Omega&\=-\sfrac 16\,\overline\eta^\ast\gamma_a\gamma_b\gamma_c\eta\,
e^a\wedge e^b\wedge e^c, \nonumber
\end{align}
where the $e^a$ form an orthonormal frame of one-forms for $T^\ast(K)$.
Every compact six-dimensional Einstein manifold carrying a nontrivial 
Killing spinor $\zeta$, i.e.\ one satisfying 
\begin{equation}
\nabla_a \zeta \= \vartheta\,\gamma_a \zeta \qquad\text{with}\qquad
\vartheta\=\pm\sfrac i2 \sqrt{\sfrac {\text{Scal}}{30}},
\end{equation}
has a nearly K\"ahler structure (the converse is also true \cite{Grunewald90}).
Indeed, one can choose $\zeta$ to be of the form $\eta +\eta^\ast$ for some 
positive-chirality spinor $\eta$, which then defines the SU(3) structure 
$(J,\omega,\Omega)$ satisfying (\ref{NK_defRels}) with the particular value
\begin{equation}
\varsigma\=\sqrt{\sfrac{\text{Scal}}{30}} \qquad\text{where}\qquad
\text{Scal}=\text{Scal}(\nabla^{LC}). 
\end{equation}
Another important quantity in the study of a nearly K\"ahler manifold is its intrinsic torsion. This is defined as the torsion of the canonical connection, and it is totally antisymmetric, which is useful in regard of the supersymmetry equations. In this paper, we identify the intrinsic torsion with (one-half of) the
$H$-flux
\begin{equation}
H \= \sfrac1{6}\,H_{abc}\,e^a\wedge e^b\wedge e^c.
\end{equation}
In terms of the basic three-form, it is given by
\begin{equation}\label{H-OmegaRel}
H\= -\sfrac{i}2\varsigma\,\big(\Omega-\overline \Omega\big)\=
\varsigma\,\text{Im}(\Omega).
\end{equation}
Therefore, the canonical connection reads
\begin{equation}
\nabla^- \= \nabla^{LC}\ +\ \sfrac12 H^a_{bc}\,e^b \otimes (E_a\otimes e^c),
\end{equation}
with vector fields $E_a$ dual to the one-forms $e^a$.

\paragraph{Properties of nearly K\"ahler manifolds.}
 Let $K$ be a nearly K\"ahler manifold. The following properties of $K$ can all be proven by some elementary spinor calculus, the assumption that $\eta+\eta^\ast$ is a Killing spinor for the Levi-Civit\`a connection on $K$, and multiple application of the following Fierz identities:
\begin{align}
\eta\ \overline\eta &\= \sfrac 18\big(1 - \sfrac 12
(\overline\eta\gamma_{ab}\eta)\gamma^{ab}\big)\big(1+(-1)^{\text{Chir}}\big), 
\nonumber \\[4pt]
\eta^\ast\overline\eta &\=-\sfrac 1{48}\Omega_{abc}\gamma^a\gamma^b\gamma^c,
\\[4pt]
\eta\ \overline\eta^\ast\!\!&\=+\sfrac 1{48}\overline\Omega_{abc}
\gamma^a\gamma^b\gamma^c. \nonumber 
\end{align}
 Here `Chir' denotes the chirality operator, giving 0 on positive chirality spinors, and 1 on negative chirality.
 Our convention for the gamma-matrices is $\{\gamma^a,\gamma^b\} = 2g^{ab}$. The forms satisfy the duality relations
\begin{equation} \label{duality}
{\ast}\Omega \= -i\Omega,\qquad 
\ast\overline\Omega \= i\overline \Omega, \qquad
2\,{\ast}\omega \= \omega\wedge \omega,
\end{equation}
and they act on the spinors $\eta$ and $\eta^\ast$ as 
\begin{equation}\label{OmegaEtaProp}
 \begin{aligned}
   \Omega_{abc}\gamma^c \eta^\ast&\=0, \qquad\qquad\qquad
   \overline\Omega_{abc}\gamma^c \eta\=0, \\[4pt]
   \Omega_{abc}\gamma^b\gamma^c\eta &\= -8\gamma_a \eta^\ast,\qquad 
   \overline\Omega_{abc}\gamma^b\gamma^c\eta^\ast \= 8\gamma_a \eta.
 \end{aligned}
\end{equation}
Our forms are normalized as
\begin{equation}
(\omega,\omega) = 3,\qquad 
(\Omega,\overline\Omega) =8,\qquad 
(\Omega,\Omega)= (\overline\Omega,\overline\Omega)=0,
\end{equation}
where $(\cdot,\cdot)$ denotes the metric induced on $\Omega(K)$ by $g$. 
This implies
\begin{equation}\label{NK_Form-Vol-Rel}
\omega^3=6\,\text{Vol} \und \Omega\wedge \overline \Omega = -8i\,\text{Vol}.
\end{equation} 
The derivatives $d\omega$ and $d\Omega$ are given in \eqref{NK_defRels}.
 
\paragraph{Coset models.}
 There are four known compact six-dimensional nearly K\"ahler manifolds, which can all be represented as coset spaces $K=G/H$, for two Lie groups $H\subset G$, where $H$ is isomorphic to a subgroup of SU(3), and not to be confused with the torsion form $H\in \Omega(K)$:
\begin{equation}
\begin{aligned}
    SU(3)/U(1)\times U(1),\qquad & \quad Sp(2)/Sp(1)\times U(1),\\[2pt]
   G_2/SU(3)=S^6,\quad\qquad & SU(2)^3/SU(2)_{\text{diag}}=S^3\times S^3.
\end{aligned}
\end{equation}

Their nearly K\"ahler structure comes from a so-called 3-symmetry, 
i.e.\ an automorphism 
\begin{equation}
s:G\rightarrow G \qquad\text{with}\qquad 
s^3=id_G \quad\text{and}\quad s\big|_H=id_H.
\end{equation} 
For a precise definition consult \cite{Butruille}. 
The differential $ds:\mathfrak g\rightarrow \mathfrak g$ 
has three possible eigenvalues, namely 
\begin{equation}
1,\qquad j=-\sfrac 12+i\sfrac{\sqrt 3}2, \qquad 
j^2=\overline j=-\sfrac 12-i\sfrac{\sqrt 3}2,
\end{equation}
with eigenspace decomposition $\mathfrak g^\mathbb C= \mathfrak h^\mathbb C \oplus \mathfrak m^+\oplus \mathfrak m^-$. Then $\mathfrak m=(\mathfrak m^ +\oplus \mathfrak m^-)\cap \mathfrak g$ can be identified with the tangent space $T_{[e]} G/H$. Identifying $\mathfrak m^+$ with (1,0)-type vectors and $\mathfrak m^-$ with (0,1)-type ones, we obtain an almost-complex structure $J$ on $G/H$, which relates to the 3-symmetry via
\begin{equation}
ds\big|_{\mathfrak m}\=-\sfrac 12 id +\sfrac{\sqrt3}2 J. 
\end{equation}
This decomposition has the following properties,
 \begin{align}
   [\mathfrak h\,,\mathfrak m^+] &\subset \mathfrak m^+,\qquad 
   [\mathfrak h\,,\mathfrak m^-] \subset \mathfrak m^-, \nonumber\\
   [\mathfrak m^+,\mathfrak m^+] &\subset \mathfrak m^- ,\qquad 
   [\mathfrak m^-,\mathfrak m^-] \subset \mathfrak m^+, \\
   [\mathfrak m^+,\mathfrak m^-] &\subset \mathfrak h^ \mathbb C, \qquad\
   [\mathfrak h\,,\mathfrak h\,] \subset \mathfrak h. \nonumber
 \end{align}

 Let $\{E_a\}_{a=1,\dots ,6}$ be a basis of $\mathfrak m$, $\{E_k\}_{k=7,\dots,6{+}\dim{\mathfrak h}}$ a basis of $\mathfrak h$, and $\{e^a\}$, $\{e^k\}$ the dual bases. As indicated, we use letters $a,b,c,\dots$ for indices in $\mathfrak m$ and $i,j,k,\dots$ for those in $\mathfrak h$. We also identify the $e^a$ and $e^k$ with the correponding left-invariant forms on $G$. Locally one can find a smooth map $\alpha:G/H\to G$, by which we pull back the forms on $G$ to forms on $G/H$, which we again denote by the same symbols. An important technical tool is the Maurer-Cartan equation
 \begin{equation}
 \begin{aligned}
 de^a &\= -\sfrac 12 f^a_{bc}\,e^b\wedge e^ c - f^a_{bk}\,e^b\wedge e^k,\\[2pt]
 de^k &\=- \sfrac 12f^k_{bc}\,e^b\wedge e^ c -\sfrac 12 f^k_{ij}\,e^i\wedge e^j.
 \end{aligned}
 \end{equation}
For the metric on $\mathfrak m$ we choose minus the Cartan-Killing form of $\mathfrak g$ restricted to $\mathfrak m$, i.e.
\begin{equation}
 g _{ab} \=- f^c_{ad}f^d_{bc} - 2f^c_{ak}f^k_{bc}.
\end{equation} 
 Both terms on the right-hand side are separately proportional to the metric
\cite{YMNK},
  \begin{equation}
    f^c_{ad}f^d_{bc} \= f^c_{ak}f^k_{bc} \= -\sfrac 13 g_{ab}.
  \end{equation} 
The last two equations fix the scale of the internal space $G/H$. 
This is done here for simplicity only; later on, we shall allow for
a rescaling~$\rho$ of the metric.
 
\paragraph{Connections and curvature in the coset models.}
 Using the above mentioned properties, one can determine the Levi-Civit\`a connection on $G/H$ by the conditions of metricity,
\begin{equation}
 0\=E_a g_{bc}\=\Gamma^d_{ab}g_{dc}+\Gamma^d_{ac}g_{bd} ,
\end{equation}
 and the vanishing of the torsion,
\begin{equation}
 de^ a +{\Gamma^a}_{b} \wedge e^b \=0.
\end{equation}
 One finds that
 \begin{equation}
   \Gamma \= \big(f^a_{ic}e^i+\sfrac12 f^a_{bc}e^ b\big) (E_a \otimes e^ c) .
 \end{equation}  
We will need not only the Levi-Civit\`a connection 
but a metric connection with torsion of the form 
\begin{equation}
T\=\kappa\,f^a_{bc}\,e^b\otimes (E_a\otimes e^c)\ \in\Omega^1(\text{End}(TK))
\qquad\text{with}\quad\kappa\in\mathbb R. 
\end{equation}
The canonical connection is included for $\kappa=-1/2$. Thus we have
\begin{equation}\label{kappaConnection}
\Gamma^\kappa \= 
\big(f^a_{ic}e^i+\sfrac12 \phi\,f^a_{bc}e^ b\big) (E_a \otimes e^ c) 
\qquad\text{with}\quad \phi:=2\kappa{+}1.
\end{equation} 
The curvature tensors of $\Gamma^\kappa$ are found to be
\begin{align}
R^\kappa &\=-\sfrac14\big\{\phi f^e_{ab}f^ c_{ed} + 2f^ k_{ab}f^ c_{kd} -\phi^2f^c_{ae} f^e_{bd}\big\}\,e^a\!\wedge e^b (E_c\otimes e^ d), \nonumber \\[2pt]
\text{Ric}^\kappa&\= -\sfrac 1{12} (4\kappa^2{-}5)\,g, \\[2pt]
\text{Scal}^\kappa &\= -\sfrac12 \phi^2+\phi+2\=-2\kappa^2+\sfrac52.\nonumber 
\end{align}  

In the following we will denote $\nabla^\kappa$ for special values 
of $\kappa$ as
\begin{equation}
\begin{aligned}
\nabla^{-\frac12}=:\nabla^- &\qquad\text{canonical connection} 
\qquad&\Longrightarrow&\qquad R^{-\frac12}=:R^-,\\[2pt]
\nabla^{0}=:\nabla\ \ &\qquad\text{Levi-Civit\`a connection} 
\qquad&\Longrightarrow&\qquad \ \ R^{0}=:R,\\
\nabla^{+\frac12}=:\nabla^+ &\qquad\text{no name} 
\qquad&\Longrightarrow&\qquad R^{+\frac12}=:R^+,\\
\end{aligned}
\end{equation}

The curvatures $R^\pm$ are sometimes denoted as $R^\pm\big|_\mathfrak m$, in 
order to distinguish them from the curvature $R^-\big|_\mathfrak h$ of another 
so-called canonical connection on the principal $H$-bundle $G\rightarrow G/H$.
The latter acts in the adjoint representation on $\mathfrak h$, but it has the 
same functional form as $R^-\big|_\mathfrak m$~\cite{KobayashiNomizu, YMNK}.
Explicit formulae are
 \begin{align}\label{CurvaturesExplicit}
  R^+ \big|_\mathfrak m&\= -\text{ad}(E_a)\circ\pi_\mathfrak h\circ \text{ad}(E_b)\,e^a\!\wedge e^b, \qquad
    {(R^+)^c}_{dab} \= 2f^c_{k[a}f^k_{b]d},\nonumber\\[2pt]
  R^-\big|_\mathfrak m&\=-\sfrac 12 f^k_{ab}\, \text{ad}_\mathfrak m (E_k)\,e^a\!\wedge e^b,\qquad\qquad\ {(R^-\big|_\mathfrak m)^c}_{dab} \=-f^k_{ab} f^c_{kd},\\[2pt]
  R^-\big|_\mathfrak h &\= -\sfrac 12 f^k_{ab}\, \text{ad}_\mathfrak h(E_k)\,e^a\!\wedge e^b,\qquad\qquad \ \ {(R^-\big|_\mathfrak h)^k}_{lab}\=-f^m_{ab}f^k_{ml }, \nonumber 
 \end{align}
 where $\pi_\mathfrak h$ is the projection of $\mathfrak g$ onto $\mathfrak h$. The torsion three-form of the canonical connection $\nabla^-$ reads
 \begin{equation}
   H\= -\sfrac 16 f_{abc}\,e^a\wedge e^b\wedge e^c.
 \end{equation} 
{}From the explicit expression for $R^\kappa$ above we compute various
curvature bilinears for later use:
 \begin{equation}
 \begin{aligned}
   \text{tr}_\mathfrak m R^+\!\wedge R^+ &\= -\sfrac 14\langle E_k,E_l\rangle _\mathfrak h f^k_{ab}f^l_{cd} e^{abcd}, \\[2pt]
   \text{tr}_\mathfrak m R^-\!\wedge R^- &\= +\sfrac 14\langle E_k,E_l\rangle _\mathfrak m f^k_{ab}f^l_{cd} e^{abcd} , \\[2pt]
   \text{tr}_\mathfrak h R^-\!\wedge R^- &\= +\sfrac 14\langle E_k,E_l\rangle _\mathfrak h f^k_{ab}f^l_{cd} e^{abcd} , \\[2pt]
   dH &\= -\sfrac 14 f_{abk}f^k_{cd} e^{abcd}, 
 \end{aligned}
 \end{equation}
where tr denotes minus the usual trace, 
$e^{abcd}=e^a\wedge e^b\wedge e^c\wedge e^d$, and 
\begin{equation}
\langle E_k,E_l\rangle_\mathfrak h \= -f^m_{kn}f^n_{lm}
\qquad\text{as well as}\qquad
\langle E_k,E_l\rangle_\mathfrak m \= -f^a_{kb}f^b_{la}
\end{equation} 
are minus the Cartan-Killing forms on $\mathfrak h$ and $\mathfrak m$, 
respectively. 
On three of the four nearly K\"ahler cosets $G/H$ 
(those where $\mathfrak h$ is semisimple) we have the identity
\begin{equation}
\langle \cdot,\cdot \rangle_\mathfrak h  \= 
\beta \langle\cdot,\cdot\rangle_\mathfrak g
\end{equation} 
for a specific value of $\beta$.
The value of $\beta$ can be calculated in orthonormal coordinates on $G$ as 
follows: 
\begin{equation}
f^k_{ab}f^k_{ab}=\sfrac 13\delta_{aa}=2 \und
f^k_{ab}f^k_{ab} =(1{-}\beta)\,\delta^{kk}=(1{-}\beta)\,\text{dim}(H)
\qquad\Longrightarrow\qquad
\beta = 1-\frac 2{\text{dim}(H)},
\end{equation}
for semisimple $\mathfrak h$.
For our examples this yields
\begin{center}
\begin{tabular}{|c|c|c|c|c|}
\hline
&SU(3)/U(1)$\times$U(1)&Sp(2)/Sp(1)$\times$U(1)&$G_2$/SU(3)&SU(2)$^3$/SU(2) \\ 
\hline 
$\beta$ & 0 & -- & 3/4 & 1/3 \\
\hline
\end{tabular}
\end{center}
and we get ($\phi=2\kappa{+}1)$
\begin{equation}\label{BianchiExpressions}
\begin{aligned}
\text{tr}_\mathfrak m R^+\!\wedge R^+ &\= \beta\, dH,  \\
\text{tr}_\mathfrak m R^-\!\wedge R^- &\= (\beta-1)\,dH, \\
\text{tr}_\mathfrak h R^-\!\wedge R^- &\= -\beta\, dH, \\
\text{tr}_\mathfrak m R^\kappa\!\wedge R^\kappa &\=
(\beta+\sfrac14\phi^2-1)\,dH,\\
\text{tr}_\mathfrak m |R^+|^2 &\= \sfrac 43-\beta, \\
\text{tr}_\mathfrak m |R^-|^2 &\= 1-\beta, \\
\text{tr}_\mathfrak h |R^-|^2 &\= \beta, \\
\text{tr}_\mathfrak m |R^\kappa|^2 &\= 
\sfrac1{24}\phi^2\big(\phi^2-2\big) +1-\beta, \\
R^+_{acde}R_{\, b}^{+cde} &\= \sfrac {4-3\beta}9\,g_{ab}, \\
R^-_{acde}R_{\, b}^{-cde} &\= \sfrac {1-\beta}3\,g_{ab}.
\end{aligned}
\end{equation}

In the compactifications we consider, 
the curvature $\tilde R$ of the tangent bundle of $K$ is one of the $R^\kappa$,
usually either $R^+$ or $R^-$. In contrast, the gauge field $F$ is free to live
on an arbitrary bundle, and so we also have the choice $F=R^-\big|_\mathfrak h$ 
at our disposal. The supersymmetry constraint, however, forces $F$ to be a 
(generalized) instanton, meaning that $\ast F= -\omega\wedge F$. 
This is satisfied only by $R^-$ (both on $\mathfrak m$ or on $\mathfrak h$). 
If $\mathfrak h$ is abelian, we will also have the freedom to rescale 
$\Gamma^-$ without losing the instanton property. Hence, for 
SU(3)/U(1)$\times$U(1) we may take $F=\lambda R^-$ with $\lambda\in\mathbb R$ 
(cf.\ \cite{Popov09}). Yet even without the supersymmetry constraint it is very
convenient to choose an instanton solution for the gauge field because it 
automatically satisfies the Yang-Mills equation. 

On the space Sp(2)/Sp(1)$\times $U(1) there is no common value for $\beta$. 
Instead, we have $\beta=0$ on $\mathfrak u(1)$ and $\beta = 2/3$ on 
$\mathfrak{sp}(1)$. This allows one to calculate the quantities 
tr$(R^\pm\wedge R^\pm)$, which are no longer proportional to $dH$. 
However, we have the freedom to restrict the curvature $\tilde R$ to the 
$\mathfrak u(1)$ part of~$\mathfrak h$.
Again this can be rescaled, enabling us to satisfy the Bianchi identity
for a particular choice, 
\begin{equation}
\tilde R=R\big|_{\mathfrak u(1)} \und 
F=R^-\big|_\mathfrak m \qquad\Longrightarrow\qquad
\text{tr}(\tilde R\wedge \tilde R - F\wedge F) \= r\, dH 
\qquad\text{with}\quad r=1/3.
\end{equation} 
In contrast to the SU(3)/U(1)$\times $U(1) case, 
there are then no free parameters left in the gauge field.
We will find below that only $r=2/3$ and $r=4/9$ are compatible with the 
dilaton equation, in the supersymmetric and non-supersymmetric case, 
respectively. Therefore we will not consider this coset space any further.

\paragraph{Forms in the coset models.}
 Now we can compare the general theory to our concrete realization 
in terms of coset models. In particular, we identified
\begin{equation}
T^c_{ab}=-\sfrac 12f^c_{ab} \und \text{Scal}\equiv\text{Scal}^0=\sfrac52 
\qquad\Longrightarrow\qquad
\varsigma \= \sqrt{\sfrac{\text{Scal}}{30}} \= \sqrt{\sfrac 1{12}}.
\end{equation}
{}From this we easily deduce
\begin{equation}
(\Omega-\overline\Omega)_{abc}\=-4\sqrt 3\,i\,f_{abc} \und H_{abc}\=-f_{abc}
\end{equation}
as well as the relations
\begin{align}
  d\omega &\= -\sfrac{3}2\varsigma\,\big(\Omega +\overline \Omega\big) \=-\sfrac{\sqrt 3}4\big(\Omega +\overline \Omega\big)\= 3\,{\ast}H,\nonumber\\[2pt]
   d\Omega &\=-d\overline \Omega \=2i\varsigma\,\omega\wedge \omega \= \sfrac 1{\sqrt 3} i\,\omega\wedge \omega, \\[2pt]
  dH &\= \sfrac1{15} \text{Scal}\,\omega\wedge \omega \= \sfrac 16\,\omega\wedge \omega.\nonumber
 \end{align}

With $\varsigma$ we have fixed the scale of the nearly K\"ahler manifold.
Let us finally see what happens if we relax this scale by allowing for an
arbitrary value of~$\varsigma$.
The new metric is then $g'=\rho g$, with $\rho=\frac 1{12 \varsigma^2}.$ 
Denoting $g'$ by $g$ again, we get the relations
\begin{equation}
\begin{aligned}
\text{Ric}&\=\sfrac 5{12}\rho^{-1} g,\qquad
\text{Scal}\=\sfrac 5{2}\rho^{-1},\\[2pt]
|H|^2 &\= \sfrac13 \rho^{-1},\qquad 
H_{acd}{H_b}^{cd} \=\sfrac13 \rho^{-1} 
g_{ab}, 
\end{aligned}
\end{equation}
and the quantities given in \eqref{BianchiExpressions} scale as follows,
\begin{equation}
\text{tr } R\wedge R\ \sim\ \rho^{-1} dH, \qquad 
\text{tr } |R|^2\ \sim\ \rho^{-2},\qquad 
R_{acde}{R_b}^{cde}\ \sim\ \rho^{-2}g_{ab}.
\end{equation}

\section{Supersymmetric solutions with gaugino condensate}\label{sec:susySols}
  In \cite{Govindarajan, Frey_Lippert} solutions of the supersymmetry equations (\ref{gauginoSusyVar}) were constructed on the product of AdS$_4$ with a nearly K\"ahler space~$K$, by setting $ \phi=0$ and $H=-\frac 14\Sigma$ equal to the intrinsic torsion of $K$. The metric on $K=G/H$ is a scale factor $\rho$, to be determined below, multiplied by (the negative of) the Killing form of $G$. Frey and Lippert in \cite{Frey_Lippert} also propose a certain ansatz for the gauge field, which is not an instanton however and therefore has a non-vanishing gaugino variation. We will discuss the choice of the gauge field later, together with the Bianchi identity and dilaton equation.
The supersymmetry generator $\varepsilon$ is obtained as follows 
\cite{Frey_Lippert}: AdS$_4(r)$ carries a Killing spinor 
$\widehat\zeta+\widehat\zeta^\ast$ with Killing number 
$\vartheta = \frac 1{2r}=\frac 1{\sqrt{2\alpha'}}$ \cite{Bohle,Baum}, i.e.
\begin{equation}\label{AdSKillingSpinor}
\nabla_\mu \widehat\zeta \= \vartheta\,\gamma_\mu\widehat\zeta^\ast \und
\nabla_\mu \widehat\zeta^\ast \= \vartheta\,\gamma_\mu\widehat\zeta.
\end{equation} 
On $K$ we have the Killing spinor $\eta$ (of positive chirality) with $\nabla^- \eta=\nabla^-\eta^\ast=0$, which gives
\begin{equation}\label{10dSpinor}
\varepsilon \=e^{\frac{i\pi}4}\widehat\zeta\otimes\eta
\ +\ e^{-\frac{i\pi}4}\widehat\zeta^\ast \otimes \eta^\ast.
\end{equation} 
Apparently the dilatino variation $\delta\lambda$ in (\ref{gauginoSusyVar}) vanishes, and for the AdS$_4$ components of $\delta\psi_\mu$ we obtain
\begin{equation}
 \begin{aligned}
   \delta\psi_\mu &\= \nabla_\mu \varepsilon\ +\ \sfrac 1{96}\gamma(\Sigma)\gamma_\mu \varepsilon \\[2pt]
      &\=e^{\frac{i\pi}4} \vartheta\,\gamma_\mu \widehat\zeta^\ast \otimes \eta\ +\ e^{-\frac{i\pi}4}\vartheta\,\gamma_\mu\widehat\zeta \otimes \eta^\ast\ -\ \sfrac 1{96}\gamma_\mu\gamma(\Sigma)\big(e^{\frac{i\pi}4}\widehat\zeta \otimes \eta + e^{-\frac{i\pi}4}\widehat\zeta^\ast \otimes \eta^\ast\big).
 \end{aligned}
\end{equation}
The condition $\Sigma=-4H$ fixes the condensate scale~$\Lambda$
in the condensate~(\ref{condensate}). Together with (\ref{H-OmegaRel}),
it follows that
\begin{equation}
\Sigma\= 2i\,\varsigma\,(\Omega-\overline\Omega) \und
\varsigma\=\sqrt{\sfrac{\text{Scal}}{30}}\=\sqrt{\sfrac{\rho}{12}}
\qquad\Leftrightarrow\qquad \rho=\frac 1{12\,\varsigma^2},
\end{equation}
relating $\varsigma$ to the scale $\rho$ of $K$. 
Together with (\ref{OmegaEtaProp}) this implies
\begin{equation}
\gamma(\Sigma)\,\eta \= -96i\,\varsigma\,\eta^\ast \und
\gamma(\Sigma)\,\eta^\ast \=-96i\,\varsigma\,\eta.
\end{equation}
Finally we get
\begin{equation}\label{GauginoPsiVrtn}
\delta \psi_\mu \= (\vartheta-\varsigma)
\Big[e^{\frac{i\pi}4}\gamma_\mu\,\widehat\zeta^\ast\!\otimes\eta + 
e^{-\frac{i\pi}4}\gamma_\mu\,\widehat\zeta\otimes \eta^\ast\Big]
\qquad\Longrightarrow\qquad \vartheta\=\varsigma.
\end{equation} 
The $K$ component $\delta \psi_a$ gives zero as well, due to $\gamma(\Sigma)\gamma_a\eta=\gamma(\Sigma)\gamma_a\eta^\ast =0$. 

\paragraph{Implications of supersymmetry.}
 Due to the condition $\vartheta=\varsigma$ we can determine the total scalar curvature of our space AdS$_4\times K$:
\begin{equation}\label{totalcurv}
\widehat{\text{Scal}}\ +\ \text{Scal} \= -48\,\vartheta^2 + 30\,\varsigma^2 \= -18\,\vartheta^2\=-\frac 9{\alpha'},
\end{equation}
which is strictly negative. Furthermore,
\begin{equation} \label{SigmaT}
\Sigma\=-4H \qquad\Longrightarrow\qquad T\=H-\sfrac 12 \Sigma \=3H, 
\end{equation}
which turns the gaugino equation (left of (\ref{trgaugino})) into the form
 \begin{equation}\label{gauginospecial}
   \big(\mathcal D -\sfrac 18 \gamma(H)\big)\chi \=0.
 \end{equation}  
This will be satisfied if we decompose the gaugino as in~(\ref{chifactor}), 
with $\eta$ being the Killing spinor of $K$ \cite{Frey_Lippert}. 
Having adjusted the value of~$\Lambda$, the three-form $\Sigma$ assumes 
the required structure $\Sigma=2i\varsigma\,(\Omega-\overline\Omega)$. 

\paragraph{Einstein and dilaton equations.}
The external Einstein equation (the first one in~(\ref{EOM_condensate}))
was analyzed in Section~\ref{sec:ActionEtc}, yielding
\begin{equation}
r^2\=\frac 1{4\,\vartheta^2} \= \frac {\alpha'}{2} \qquad\Longrightarrow\qquad
\varsigma^2\= \frac 1{2\alpha'} \und \rho\=\frac{\alpha'}6.
\end{equation}
Next, we want to consider the fourth equation in~(\ref{EOM_condensate}), 
which can be written as
\begin{equation}\label{EinsteinDilatonCombiGaugino:SUSY}
\widehat{\text{Scal}}\ +\ \text{Scal}\ +\ \sfrac 92 |H|^2\=0.
 \end{equation} 
The normalization of the intrinsic torsion $H$ was such that 
$|H|^2=4\varsigma^2$, 
and together with (\ref{totalcurv}) and (\ref{GauginoPsiVrtn}),
we deduce that (\ref{EinsteinDilatonCombiGaugino:SUSY}) is satisfied.

Let us then determine the remaining conditions for $\tilde R$ and $F$ 
following from the equations of motion. Due to (\ref{SigmaT}), 
the third equation in~(\ref{EOM_condensate}) becomes
\begin{equation}\label{gaugino:DilatonEqa}
\sfrac14\alpha'\big[\sfrac {12}{r^4}+
\text{tr}\big(|\tilde R|^2 - |F|^2 \big)\big] \= 
3|H|^2 \=12\varsigma^2,
\end{equation} 
and the Bianchi identity remains unchanged:
\begin{equation}
dH\= \sfrac14\alpha' \text{tr}
\big[ \tilde R\wedge \tilde R - F\wedge F\big].
\end{equation} 
Written in terms of $\rho=\frac 1{12\varsigma^2}=\frac {\alpha'}6$, 
we arrive at the conditions
\begin{equation}\label{dilBianchiSUSY}
\text{tr}\big[|\tilde R|^2 -|F|^2 \big] \=- \sfrac 23\rho^{-2} \und
\text{tr}\big[\tilde R\wedge \tilde R-F\wedge F\big]\=
\sfrac 23\rho^{-1} dH. 
\end{equation}
Further we need the Einstein equation on $K$ 
(the second one in~(\ref{EOM_condensate})), 
which must imply (\ref{gaugino:DilatonEqa}). 
It is easy to see that all terms in this equation are proportional to $g_{ab}$,
so that it becomes equivalent to its trace (\ref{gaugino:DilatonEqa}). 

\paragraph{Choice of connections.}
The gaugino equation (left of (\ref{trgaugino})) is satisfied because
$\eta$ is parallel with respect to $\nabla^-$. The Kalb-Ramond equation
is solved due to (\ref{NK_defRels}), (\ref{H-OmegaRel}) and (\ref{duality}).
The Yang-Mills equation will be obeyed for (generalized) instanton connections.
The nearly K\"ahler geometry of~$K$ determines~$H$, but for $\tilde R$ and $F$
different choices are possible, the consequences of which can be computed 
using~\eqref{BianchiExpressions}. 
To this end, the following lemma will be useful:

\begin{lem}\label{InstantonBianchiDilRel}
Let $F$ be a (generalized) instanton, 
i.e.\ $F$ satisfies the $\omega$-anti-self-duality relation
$$ \ast F\= - \omega \wedge F,$$
(which is equivalent to the vanishing of the gaugino supersymmetry variation). 
Suppose further that 
$$\text{tr}\,(F\wedge F) \= \frac \varkappa \rho\, dH.$$ 
Then it follows that
$$\text{tr}\,|F|^2 \= -\frac \varkappa{\rho^2}.$$
\end{lem}
\begin{proof}
We have
\begin{equation*}
\text{tr}\,|F|^2 \= \text{tr}\ast\!(F\wedge \ast F) 
          \= -\text{tr}\ast\!(F\wedge \omega \wedge F) 
          \= -\frac\varkappa\rho \ast\!\big(\omega \wedge dH \big) 
          \= -\frac \varkappa{6\rho^2} \ast\!\omega^3\=-\frac\varkappa{\rho^2},
\end{equation*}
where the relations $\omega^3=6\,$Vol and $dH=\frac 1{6\rho}\omega\wedge\omega$ 
for nearly K\"ahler spaces in our particular normalization have been used. 
\end{proof}
The lemma implies that if both $\tilde R$ and $F$ are instantons, 
with tr$(F\wedge F)\sim dH$, then the two equations \eqref{dilBianchiSUSY} 
become equivalent to one another, and we need to solve only one of them. 
Hence, in this case the vanishing of the supersymmetry constraints plus the 
Bianchi identity imply the field equations, as expected.
 
The gauge field $A$ proposed by Frey and Lippert is the torsionful connection 
$\Gamma^\kappa$ as in \eqref{kappaConnection}, where the value of $\kappa$ 
has to be adjusted to the Bianchi identity. However, they employ
$\tilde\Gamma=\Gamma^+$.\footnote{
Note that our $\Gamma^+$ is their $\Gamma^-$.} 
The problem with this choice is that $\Gamma^\kappa$ 
is in general not an instanton, leading to a non-vanishing gaugino variation. 
We will therefore consider $\tilde\Gamma=\Gamma^-$ and other possibilities
of the gauge connection as well.

\paragraph{The minus-connection: $\mathbf{\tilde \Gamma=\Gamma^-}$.}
First we consider solutions where $\tilde R=R^-\big|_\mathfrak m$ 
and $F$ is also an instanton, i.e.\ $F=R^-\big|_\mathfrak h$ or, 
for $\mathfrak h$ abelian, $F=\lambda R^-\big|_\mathfrak m$ with 
$\lambda\in\mathbb R$. The Bianchi identity reads
\begin{equation}
\text{tr}\, F\wedge F \= \big(\beta - \sfrac 53\big)\rho^{-1}dH.
\end{equation} 
Choosing the canonical $H$-connection 
\begin{equation}
F\=R^-\big|_\mathfrak h \qquad\text{with}\qquad 
\text{tr}\,(F\wedge F)\= -\sfrac\beta\rho\,dH
\qquad\Longrightarrow\qquad \beta=\sfrac56, 
\end{equation}
which is not among the admissible values. 
A solution is obtained however on SU(3)/U(1)$\times $U(1) $(\beta=0$), 
where we can choose 
\begin{equation}
F\=\sqrt{\sfrac 53}\,R^-\big|_\mathfrak m \qquad\Longrightarrow\qquad
\text{tr}\,(F\wedge F)\=-\sfrac 53\,dH.
\end{equation} 
The explicit expressions for $F$ and $\tilde R$ can be found 
in~\eqref{CurvaturesExplicit}.

\paragraph{The plus-connection: $\mathbf{\tilde \Gamma=\Gamma^+}$.}
Now we turn to $\tilde R=R^+$ solutions, so that $\tilde R$ is not an instanton
any more. Here, tr($R^+\!\wedge R^ +)$ and tr$\,|R^ +|^2$ have been calculated 
in~\eqref{BianchiExpressions} and lead to the conditions 
\begin{equation}
\begin{aligned}
\text{tr}\,F\wedge F&\=\big(\beta-\sfrac23\big)\rho^{-1}dH 
\qquad\ (\text{Bianchi}), \\[2pt]
\text{tr}\,|F|^2 &\= (2-\beta)\,\rho^{-2} 
\qquad\qquad (\text{Dilaton}), 
\end{aligned}
\end{equation}
in contradiction with Lemma \ref{InstantonBianchiDilRel}. 
The Bianchi identity on SU(2)$^3$/SU(2) with $\beta=1/3$ can be obeyed for
$F=R^-\big|_\mathfrak h$, but this choice fails to solve the dilaton equation
(second of (\ref{EinsteinDilatonCombiGaugino})). 
Thus we confirm Ivanov's result that only for 
$\tilde\Gamma=\Gamma^-$ the field equations follow. 
The string theory formulated in this background then has a conformal anomaly, 
but is otherwise consistent \cite{Polchinski}.
The situation can be remedied by turning on a dilatino condensate.
This introduces the freedom necessary to generate consistent solutions with 
$\tilde\Gamma=\Gamma^+$ on all nearly K\"ahler spaces~\cite{Manousselis:2005xa}.

\section{Non-supersymmetric solutions with vanishing gaugino}
\label{sec:nonSUSYNKCptf}

 We will now construct a solution of the bosonic equations of motion on AdS$_4\times K$ with vanishing gaugino, which necessarily breaks supersymmetry. 
As before we choose the $H$-field equal to the intrinsic torsion. 

\paragraph{Equations of motion for $\mathbf{\chi=0}$.}
 For $\phi=0$ and $\chi=0$ the equations \eqref{EinsteinDilatonCombiGaugino} simplify to
\begin{equation}
 \begin{aligned}\label{non-susy-dilaton}
|H|^2 &\=\sfrac14\alpha'\text{tr}\big[\sfrac{12}{r^4}+|\tilde R|^2-|F|^2\big]
\qquad\Longrightarrow\qquad \sfrac{48}{{\alpha'}^2}\=  
\sfrac4{\alpha'}|H|^2 +\text{tr}\,|F|^2-\text{tr}\,|\tilde{R}|^2,\\[2pt]
     0&\= -\sfrac{24}{r^2}+2\,\text{Scal} +|H|^2 
\qquad\quad\ \,\Longrightarrow\qquad\ \,
\sfrac{24}{\alpha'}\=\text{Scal} +\sfrac 12 |H|^2
\qquad\Longrightarrow\qquad \rho=\frac{\alpha'}9,
 \end{aligned}
\end{equation}
by virtue of $r^2=\frac{\alpha'}2$. The value for $\rho$ differs from the
supersymmetric one~(\ref{totalcurv}).
The Bianchi identity and the remaining equations of motion become
\begin{equation}\label{RemEqns2}
\begin{aligned}
dH&\=\sfrac14\alpha'\text{tr}\big[\tilde R\wedge\tilde R-F\wedge F\big],\\[4pt]
0&\=d\ast\!H,\\[4pt]
0&\=d F + A\wedge \ast F- \ast F\wedge A  +\ast H\wedge F,\\
0&\= \text{Ric}_{ab} -\sfrac14 H_{acd}{H_b}^{cd} +\sfrac14\alpha' 
\big[\tilde R_{acde}{\tilde R_{b}}^{\ cde}-\text{tr}F_{ac}{F_b}^c\big]\\[2pt]
&\= \sfrac13 g_{ab} +\sfrac14\alpha'\big[ 
\tilde R_{acde} {\tilde R_{b}}^{\ cde} - \text{tr}F_{ac} {F_b}^c\big],
\end{aligned}
\end{equation}
Since all terms in the Einstein equation are proportional to~$g_{ab}$, 
its information is already contained in the trace,
\begin{equation}\label{non-susy-dilaton2}
\frac {36}{(\alpha')^2} \= \text{tr}\Big[|F|^2 -|\tilde R|^2\Big],
\end{equation}
which is just the first equation of \eqref{non-susy-dilaton}. We will solve the Yang-Mills equation simply by imposing the instanton condition on $F$, so that all equations are satisfied, except for the Bianchi identity and the dilaton equation \eqref{non-susy-dilaton2}. 
 Let us again assume that both $F$ and $\tilde R$ satisfy the instanton condition, as well as 
\begin{equation}
\tr\,(\tilde R\wedge R)\=\frac {\varkappa_1}{\rho}\, dH \und 
\tr\,(F\wedge F) \= \frac{\varkappa_2}{\rho }\, dH. 
\end{equation}
Using lemma \ref{InstantonBianchiDilRel}, we can rewrite these two conditions as
\begin{equation}
\Big( \frac {\varkappa_1}\rho-\frac {\varkappa_2}\rho\Big) \= \frac 4{\alpha'} 
\und \Big( \frac {\varkappa_1}{\rho^2}-\frac {\varkappa_2}{\rho^2}\Big) \= 
\frac {36}{(\alpha')^2}. 
\end{equation}

\paragraph{The minus-connection: $\mathbf{\tilde\Gamma=\Gamma^-}$.}
Let us try to fulfil the equations using $\Gamma^-$. 
According to the discussion above, a single equation remains to be solved:
\begin{equation} \label{single}
\text{tr }F\wedge F \=\big(\beta -\sfrac{13}{9} \big)\rho^{-1}  dH. 
\end{equation}
The most obvious choice is the canonical $H$-connection, 
i.e.\ $\Gamma^-\big|_\mathfrak h$. 
This however has tr$_{\mathfrak h}F\wedge F=-\beta\,dH$, 
so that (\ref{single}) implies $\beta=\frac{13}{18 }$, which is not admissible.
Only on SU(3)/U(1)$\times $U(1) we achieve a solution for 
$F=\sqrt{\sfrac{13}9}\,R^-\big|_\mathfrak m$. 

\paragraph{The plus-connection: $\mathbf{\tilde\Gamma=\Gamma^+}$.} 
In this case, the Bianchi identity and Einstein equation read
\begin{equation}
\text{tr } F\wedge F \= \big(\beta-\sfrac 49\big)\rho^{-1}dH \und
\text{tr} |F|^2\=\sfrac {16-9\beta}{9\rho^2}.
\end{equation}
These equations do not admit instanton solutions for $F$, as 
\begin{equation}
\text{tr } F\wedge F \= \big(\beta-\sfrac 49\big)\rho^{-1}dH
\qquad\buildrel\text{Lemma \ref{InstantonBianchiDilRel}}\over\Longrightarrow
\qquad \text{tr } |F|^2\= \big(\sfrac 49-\beta\big)\rho^{-2},
\end{equation}
which leads to the contradiction $\frac49-\beta=\frac{16}9-\beta$.

Furthermore, there are no models with a conformal anomaly, 
satisfying only the Bianchi identity.
Non-instanton solutions are not excluded by this argument, but these have 
the drawback that the Yang-Mills equations have to be checked explicitly.  

Another possibility is to allow for more general torsionful connections, 
$\tilde\Gamma=\Gamma^\kappa$, and maybe also $F=\Gamma^{\kappa'}$. 
This does not lead to further solutions however, as one can see from the 
relevant expressions calculated in~\eqref{BianchiExpressions}.

\paragraph{The volume modulus.} 
One may try to deform our supersymmetric solution to obtain a family
of solutions, possibly relaxing the supersymmetry constraint. A simple option
is a rescaling
\begin{equation}
\chi' =\tau \chi, \qquad g'=\rho\,g \und F'=\theta F.
\end{equation}
In this case, the dilaton and Einstein equations together with the Bianchi 
identity imply 
\begin{equation}
\tau=0 \qquad\text{or}\qquad \tau=1 \qquad\Longrightarrow\qquad
\alpha'=9\rho \qquad\text{or}\qquad \alpha'=6\rho
\end{equation}
and the corresponding values for $F$ we found above. 
In particular, our two solutions to the equations of motion cannot be 
continuously connected by an obvious path in the field space.

\section{Conclusions}
 
 We found that for two nearly K\"ahler spaces, SU(2)$^3$/SU(2) and SU(3)/U(1)$\times$U(1), the supersymmetric model of Frey and Lippert can be completed by an appropriate choice of gauge field, and that the equations of motion are satisfied on the second space, whereas the first one has a conformal anomaly at order $\alpha'$, i.e. the dilaton equation (fourth of \eqref{EOM_split}) does not hold. This is due to the choice  $\tilde \Gamma=\Gamma^+$ required by the SU(2)$^3$/SU(2)-model, whereas on SU(3)/U(1)$\times$U(1) we can work with the instanton connection $\tilde \Gamma=\Gamma^-$. 
The explicit form of our solution on 
AdS${}_4(r)\,\times$ SU(3)/U(1)$\times$U(1) 
with radius $r=\sqrt{\alpha'/2}$ and internal scale $\rho=\alpha'/6$ reads
\begin{equation}
\begin{aligned}
g_{ab} &\= -\sfrac12\alpha' f_{ad}^c f_{bc}^d \= 
-\sfrac12\alpha' f_{ak}^c f_{bc}^k, \\[2pt]
(\tilde\Gamma)_a^c &\= (\Gamma^-)_a^c \= f_{ic}^a\,e^i ,\\[2pt]
T_{abc} &\= 3\,H_{abc} \= -3\,f_{abc}, \\
F_{ab} &\= \sqrt{\sfrac53}\,R_{ab}^-\big|_{\mathfrak m} \=
-\sqrt{\sfrac 53}\,f^k_{ab}\ \text{ad}(E_k), \\[2pt]
\phi &\= \psi \= \lambda \= 0, \\[2pt]
\chi &\= \widehat\chi\otimes\eta\ +\ \widehat\chi^*\otimes\eta^*
\und \langle\Sigma_{abc}\rangle \= -\sfrac43\,T_{abc} \= 4\,f_{abc}. 
\end{aligned}
\end{equation}

On the same space (but with $\rho=\alpha'/9$) 
we were also able to construct a solution to the equations 
of motion with non-vanishing fermionic supersymmetry variations,
\begin{equation}
\begin{aligned}
g_{ab} &\= -\sfrac13\alpha' f_{ad}^c f_{bc}^d \= 
-\sfrac13\alpha' f_{ak}^c f_{bc}^k, \\[2pt]
(\tilde\Gamma)_a^c &\=(\Gamma^-)_a^c \= f_{ic}^a\,e^i, \\[2pt]
T_{abc} &\= H_{abc} \= -f_{abc}, \\
F_{ab} &\= \sqrt{\sfrac{13}9}\ R_{ab}^-\big|_{\mathfrak m} \=
-\sqrt{\sfrac{13}9}\,f^k_{ab}\ \text{ad}(E_k), \\[2pt]
\phi &\= \psi \= \lambda \= \chi \= 0 
\qquad\Longrightarrow\qquad \langle\Sigma\rangle\=0.
\qquad\qquad\qquad\quad
\end{aligned}
\end{equation}

It was known previously \cite{Manousselis:2005xa} that there exist complete solutions of the supersymmetry constraints and Bianchi identity (with $\tilde \Gamma=\Gamma^+$) on every nearly K\"ahler manifold if one includes a dilatino condensate, but the derivation of the equations of motion for this case is more involved. In the limit of vanishing dilatino condensate, these solutions still differ from ours.

A nice feature of the space-time Einstein equation with $\alpha'$ correction, 
\begin{equation}
\text{Ric}_{\mu\nu}\ +\ \sfrac14\alpha' 
R_{\mu\alpha\beta\gamma}{R_\nu}^{\alpha\beta\gamma}\=0,
\end{equation} 
is that it fixes the radius of AdS$_4$ (in terms of $\alpha'$) and thereby, in combination with the vanishing of the gravitino variation, also the scale of the internal manifold. There is thus no volume modulus in the game, an argument which apparently does not apply to Minkowski compactifications. 
It is not the supersymmetry constraints which fix the AdS$_4$ radius; even our non-supersymmetric solution seems to have no volume modulus. However, it is not clear how the relation between $r$ and $\alpha'$ behaves under the inclusion of higher-order $\alpha'$ corrections.

What concerns the connection $\tilde \Gamma$ we confirmed Ivanov's result \cite{Ivanov} that supersymmetry constraints, Bianchi identity, and equations of motion are compatible only for $\tilde \Gamma=\Gamma^-$, even in the case of AdS$_4$ compactifications with a gaugino condensate. Even without supersymmetry, we found that the choice $\tilde \Gamma=\Gamma^-$ is mandatory.

\paragraph{Acknowledgments.}
This work was partially supported by the Deutsche Forschungsgemeinschaft,
the cluster of excellence QUEST, the Heisenberg-Landau program and 
RFBR grant 09-02-91347.

\end{document}